\documentclass{amsart}

\usepackage{amsmath}
\usepackage{amsthm}
\usepackage{amsfonts}
\usepackage{xcolor}

\newtheorem{proposition}{Proposition}
\newtheorem{corollary}{Corollary}
\newtheorem{lemma}{Lemma}
\newtheorem{definition}{Definition}

\def\R{\mathbb R}
\def\Z{\mathbb Z}
\def\calP{\mathcal P}

\def\half{\frac 1 2}

\begin{document}

\title{Long cycles in linear thresholding systems}

\author{Anna Laddach}
\address{Francis Crick Institute, London}

\author{Michael Shapiro}
\address{Francis Crick Institute, London}

\date{22 November 2023}

\maketitle

\section*{Introduction}

Linear thresholding systems have been studied as models of neural
activation \cite{coolen} and more recently as models of cell-intrinsic
gene regulation \cite{ 2017JPhA...50P5601H}.  More generally, they
represent dynamics on a set of $N$ objects which can be either on or
off.  That is, they are maps
\begin{align*}
  & f = f_{(J,\theta)} : \calP_N = \{0,1\}^N \subset \R^N \to  \{0,1\}^N  \\
  & f(x) = \theta(Jx)
\end{align*}
where
$J$ is a linear map and $\theta$ is a thresholding function
\begin{align*}
  & \theta(x_1,\dots,x_N) = (y_1,\dots,y_N) \\
  & y_i = \begin{cases}
    0 & \text{if } x_i \le \theta_i \\
    1 & \text{if } x_i > \theta_i \\
  \end{cases}
\end{align*}
We will restrict to the case in which each $\theta_i$ is positive.
Since there are $2^N$ elements of $\calP_N$, this provides an upper
bound on the length of any cycle for a linear thresholding system in
dimension $N$.  Under the assumption that the $\theta_i$ are positive,
0 is a fixed point and this reduces this bound to  $2^N - 1$.  Here we
exhibit systems of arbitrarily large $N$ with cycles of length greater
than $e^{\sqrt{N}}$.  Along the way, we introduce a direct product on
linear thresholding systems.

\section*{Background and definitions}

We record here some basic definitions and results.

\begin{definition}
  ~
  \begin{itemize}
  \item We say two linear thresholding systems $(J,\theta)$ and
    $(J',\theta')$ are {\em equivalent} and write  $(J,\theta) \sim
    (J',\theta')$ if $f_{(J,\theta)} = f_{(J',\theta')}$.
  \item We say that $(J,\theta)$ is {\em generic} if for each $x \in
    \calP$, and each cooridinate $i$, $(Jx)_i \ne \theta_i$.
  \item We will say that $(J,\theta)$ is {\em robust} if there is an
    open set $U_J$ around $J$ and an open set $U_\theta$ around $\theta$
    such that for each $(J',\theta') \in U_J \times U_\theta$,
    $(J,\theta) \sim (J',\theta')$
  \end{itemize}
  
\end{definition}

The following are not hard to see:

\begin{proposition}
  ~
  \begin{itemize}
     \item $(J,\theta)$ is robust if and only if it is generic.
     \item Every linear thresholding system is equivalent to a generic
       system.
     \item Every linear thresholding system is equivalent under a
       linear change of coordinates to a system in which the
       $\theta_i$ are independent of $i$. Further, their value can be
       taken to be any positive number.
  \end{itemize}
\end{proposition}

We will take our thresholding value to be $\theta=\half$.  We will be
interested in cases where $J$ permutes $\calP$, and here thresholding
acts as the identity.

The dynamics of a function $f$ on a finite set $X$ can be seen as a
directed graph $\Gamma(X,E)$ where the edges $E$ are the pairs
$\{(x,f(x)) \mid x \in X \}$.  Each connected component of $\Gamma$
consists of a recurrent cycle $C$ (possibly of length 1) possibly
decorated with directed trees $T_i$.  Each $T_i$ is directed towards
its root, which is a vertex of $C$.  The vertices on the cycles are
{\em recurrent}, the others are {\em transient}.

\section*{Long cycles}

We start by defining a direct product on linear thresholding
systems.

\begin{definition}
Given two systems, $(J,\theta)$ and $(J',\theta)$ where the first has
$N$ genes and the second has $N'$ genes, the {\em product system},
$(J,\theta)\oplus(J',\theta)$ is the system $(\bar J,\theta)$ on
$N+N'$ genes where
$$ \bar J = \left[ \begin{matrix}
    J & 0 \\ 
    0 & J'\end{matrix} \right].$$
\end{definition}

\begin{proposition}
The dynamics of the product system are the product of the dynamics of
the two systems.  That is,  let $f$ and $f'$ be the functions of
the two systems and $\bar f$ be the function of the product system.
Let $x$ and $x'$ be elements of $\calP(1,\dots,N)$ and
$\calP(1,\dots,N')$.  Then 
$$\bar f\left(\left[\begin{matrix} x\\ x'
\end{matrix}\right]\right) = 
\left[\begin{matrix} f(x) \\ f'(x') \end{matrix} \right]
$$ \qed
\end{proposition}

\begin{corollary}
Suppose $x_1,\dots,x_r$ is a cycle of $(J,\theta)$ of length $r$ and
$x'_1,\dots,x'_s$ is a cycle of $(J',\theta)$ of length $s$.  Then
together, they form a cycle of length $\mathop{\mathrm {lcm}}(r,s)$ in
the product system. \qed
\end{corollary}

We now take a moment to consider the orbits of $(J_p,\theta)$ where
$J=J_p$ is the matrix that cyclically permutes the bases of $\R^p$, $p$
is prime and $\theta=\frac{1}{2}$.  Observe that in this case the
thresholding operation is superfluous since it acts as the identity on
$J \calP$.  Note that $J$ generates the cyclic group $\Z_p$ and
thereby induces an action of $\Z_p$ on $\calP$.

\begin{lemma} ~
  \begin{itemize}
  \item $(J,\theta)$ has no transient states.
  \item The number of coordinates which are one is constant on each orbit.
  \item Each cycle of $(J,\theta)$ has length 1 or $p$.
  \item There are two cycles of length 1.
  \item There are $(2^p - 2) / p$ cycles of length $p$.
  \end{itemize}
\end{lemma}

\begin{proof} ~
  \begin{itemize}
  \item Since $J^p = I$, $J^p x = x$ for each $x \in \calP$.
  \item Each $x \in \calP$ is a set coordinates which are one.  $J$ acts
    to cyclically permute this set.  It follows that the number of
    coordinates which are one is constant along the orbit.
  \item The length of each orbit divides $p$.  Since $p$ is prime, it
    follows that this is necessarily 1 or $p$.
  \item The points $(0,\dots,0)$ and $(1,\dots,1)$ are fixed by $J$.
    On the other hand, if $x$ is neither of these, $x$ has two
    (cyclically) adjacent coordinates which are not equal.  Hence $Jx
    \ne x$, so the orbit of $x$ has length $p$.
  \item This follows directly from the previous statement.
  \end{itemize}
\end{proof}

\begin{corollary} \label{longCycles}
Let $p_1,\dots,p_k$ be a sequence of distinct primes, $N = \sum p_i$,
$P = \prod p_i$.  Then there is a system $(J,\theta)$ on $N$ genes with a
cycle of length $P$.  \qed
\end{corollary}

Note that the sum of the first $n$ primes grows as $\frac 1 2
n^2\log(n)$ \cite{primeSum} while the product of the first $n$ primes
grows as $e^{(1+o(1)) n \log(n)}$ \cite{ruiz}.  Thus, for arbitrarily
large values $N$ there are systems on $N$ genes with cycles of length
greater than $e^{\sqrt{N}}$.

In fact, $(J,\theta)$ is rich in long cycles.  If $S \subset
\{1,\dots,k \}$, then there are
$$ \prod_{i \in S} (2^{p_i} - 2) / p_i $$
  orbits of length $\prod_{i \in S} p_i$.

Our definition is restricted to the case where the two systems have the
same thresholding function $\theta$.  However up to equivalence, they
can be given the same $\theta$ and the equivalence class of the
product system depends only on the equivalence class of the factors.
Thus, up to equivalence, the product is defined on arbitrary pairs of
systems.

\bibliographystyle{plain}
\bibliography{linearThresholding}

\begin{thebibliography}{1}

\bibitem{coolen}
{A C C} Coolen, P~Sollich, and R~Kuehn.
\newblock {\em Theory of neural information processing systems}.
\newblock Oxford University Press, 2005.

\bibitem{2017JPhA...50P5601H}
Ryan {Hannam}, Alessia {Annibale}, and Reimer {K{\"u}hn}.
\newblock {Cell reprogramming modelled as transitions in a hierarchy of cell
  cycles}.
\newblock {\em Journal of Physics A Mathematical General}, 50(42):425601,
  October 2017.

\bibitem{primeSum}
E.~Landau.
\newblock {\em Handbuch der Lehre von der Verteilung der Primzahlen}, volume~1.
\newblock Teubner, Leipzig und Berlin, 1909.

\bibitem{ruiz}
Sebastian~Martin Ruiz.
\newblock A result on prime numbers.
\newblock {\em Math. Gaz.}, July 1997.

\end{thebibliography}

\end{document}